\documentclass[notitlepage]{article}

\usepackage{biblatex}

\usepackage{fullpage}
\usepackage{titling}
\usepackage{authblk}
\usepackage[hidelinks]{hyperref}
\usepackage[dvipsnames]{xcolor}
\usepackage{comment}

\usepackage{amsmath}
\usepackage{amssymb}
\usepackage{amsthm}
\usepackage{mathtools}
\usepackage{bm}

\newcommand\abs[1]{\lvert#1\rvert}

\newtheorem{theoremcounter}{Theorem}
\newtheorem{proposition}[theoremcounter]{Proposition}
\newtheorem{lemma}[theoremcounter]{Lemma}

\newtheorem{remark}[theoremcounter]{Remark}

\def\Def{\vcentcolon=}

\begin{document}

\title{Short Proof: Exact Solution to the Finite Frobenius Coin Problem}
\date{}

\author[1]{Lorenzo De Gaspari}
\author[2]{Marco Ronzani}
{
    \makeatletter
    \renewcommand\AB@affilsepx{: \protect\Affilfont}
    \makeatother
    \affil[ ]{Affiliations}
    \makeatletter
    \renewcommand\AB@affilsepx{, \protect\Affilfont}
    \makeatother
    \affil[2]{DEIB, Politecnico di Milano}
}
{
    \makeatletter
    \renewcommand\AB@affilsepx{: \protect\Affilfont}
    \makeatother
    \affil[ ]{\authorcr Contacts}
    \makeatletter
    \renewcommand\AB@affilsepx{, \protect\Affilfont}
    \makeatother
    \affil[1]{\href{mailto:degaspari.lorenzo.edu@gmail.com}{degaspari.lorenzo.edu@gmail.com}}
    \affil[2]{\href{mailto:marco.ronzani@polimi.it}{marco.ronzani@polimi.it}}
}

\maketitle

\begin{abstract}
    The Frobenius Coin Problem is a classic question in mathematics: given coins of specified denominations, what is the largest amount that cannot be formed using only those coins?
    This brief work covers a variation of such question, posing a limit on the number of coins available for each denomination.
    Thus, the new problem becomes finding the count of distinct values that can be represented, and those that cannot, within the finite set of integers ranging from zero to the sum of all coins.
    We refer to this version of the problem as the ``finite'' case.
    We will show how this closely relates to the original question, and prove an exact formula solving the problem when exactly two denominations are involved.
\end{abstract}

\vspace{28pt}


\section{Problem Statement}

Let $A$, $B$, $m$, and $n$ be positive integers such that $m$ and $n$ are co-prime.
We want to determine the total count of distinct integer values that can be achieved by the restricted integer linear combination:
\begin{equation*}
    am + bn
\end{equation*}
as $a \in \{0, \ldots, A\}$ and $b \in \{0, \ldots, B\}$ vary.

This is a variant of the famous Frobenius Coin Problem for two numbers \cite{DiophantineFrobeniusProblem}, in which the integer linear combination is instead unrestricted (that is, $a, b \in \mathbb{N}$).
We recall that, as proven by J. J. Sylvester's solution~\cite{SylvesterFrobeniusProof}, in such original version there are exactly
\begin{equation}\label{def:h}
    h(m, n) \Def \frac{1}{2}(m - 1)(n - 1)
\end{equation}
unachievable non-negative integer values of which the largest is
\begin{equation}\label{def:g}
    g(m, n) \Def mn - m - n \text{,}
\end{equation}
that is sometimes called "Frobenius number". 
Since we consider $m$ and $n$ to be fixed constants, from now on we will be omitting the dependence of $h$ and $g$ on such quantities.

\section{Discussion}

First, we define the set of admissible coefficient pairs:
\begin{equation}\label{def:I}
    I \Def \{0, \ldots, A\} \times \{0, \ldots, B\} \text{.}
\end{equation}
Let us call ``valid'' any pair $(a, b) \in \mathbb{Z}$ s.t. $(a, b) \in I$.

Then, let us redefine the problem as finding the cardinality of the set: 
\begin{equation}\label{def:S}
    S \Def \left\{ am + bn \colon (a, b) \in I \right\} \text{,}
\end{equation}
that is, finding $\abs{S}$.

\begin{proposition}\label{prop:main_result}
Let $A$, $B$, $m$, and $n$ be positive integers such that $m$ and $n$ are co-primes.
Moreover, let $I$ and $S$ be defined as in \eqref{def:I}, \eqref{def:S}.
Then:
\begin{equation}
    \abs{S} =
    \begin{cases}
        Am + Bn + 1 - 2h & \text{if } A \geq n \text{, } B \geq m \\
        (A + 1)(B + 1) & \text{otherwise.} \\
    \end{cases}
\end{equation}
\end{proposition}

We will proceed to prove the two above formulae.

\subsection{Small Coefficients Case}

This addresses the ``otherwise'' case.
Intuitively, when $A$ and $B$ are smaller than $n$ and $m$ respectively, the latters' valid linear combinations become all distinct, thus trivially counting them answers the problem.

\begin{lemma}\label{lemma:small_coeffs}
    If either $A < n$ or $B < m$ then $\abs{S} = (A + 1)(B + 1)$.
\end{lemma}

\begin{proof}
    Clearly, $\abs{I} = (A + 1)(B + 1)$.
    Consider a function $f \colon I \rightarrow S$ s.t. $f((a, b)) \Def am + bn$. it is surjective by the definition of $S$.
    If $f$ were to also be injective, then we would have $\abs{I} = \abs{S}$, concluding the proof.
    
    Suppose by contradiction that $f$ is not injective, then $f((a_1, b_1)) = f((a_2, b_2))$ for some $(a_1, b_1), (a_2, b_2) \in I$, $(a_1, b_1) \neq (a_2, b_2)$.
    Then, $(a_1 - a_2)m = (b_2 - b_1)n$, where both the right- and left-hand side are integers, and respectively multiples of $m$ and $n$, thus by the equality both are multiples of $mn$ due to the hypothesis of $m$ and $n$ being co-primes.
    Hence, $(a_1 - a_2)$ is a multiple of $n$ as well.
    By the definition of $I$ and the hypothesis $A < n$ we get $0 \leq a_1 - a_2 \leq A < n$, it follows that $a_1 - a_2 = 0$.
    Analogously, $b_2 - b_1 = 0$, which implies the injectivity of $f$ and concludes the proof.
\end{proof}

\subsection{Large Coefficients Case}

In this more general case, where $A$ and $B$ are larger than $n$ and $m$ respectively, there could be multiple valid linear combinations of $n$ and $m$ producing the same value within $S$.
Therefore, our approach will be to leverage the result in \eqref{def:h} to remove the count of unachievable values from $\abs{I}$, subsequently proving that such removed values were all the unachievable ones.

Let us define the set of unachievable values
\begin{equation}\label{def:H}
    H \Def \{0, \ldots, Am + Bn\} \,\smallsetminus\, S\text{.}
\end{equation}
Consider now the original Frobenius problem and let $H_o$ be its set of unachievable values, namely
\begin{equation}\label{def:H_o}
    H_o \Def \{v_o \in \mathbb{N} \colon am+bn \neq v_o \;\forall (a,b) \in \mathbb{N}^2\}\text{,}
\end{equation}
and $\abs{H_o} = h$ by \eqref{def:h}.
Then surely $H_o \subseteq H$ since $g < 2mn \leq Am + Bn$, where $g$ (see \eqref{def:g}) is the largest value in $H_o$, and restricting the linear combination may only introduce additional unachievable values.

\begin{remark}\label{rem:reflection}
    The problem, and $S$ in particular, enjoys a reflection property.
    Given a pair $(a, b) \in I$, clearly both $s \Def am + bn$ and $t \Def (A - a)m + (B - b)n$ are in $S$.
    And we can interpret $t$ as the refection of $s$, that is, $t = (Am + Bn) - s$.
\end{remark}

\begin{lemma}\label{lemma:large_coeffs_like_original}
    If both $A \geq n$ and $B \geq m$ then $\abs{S} \leq Am + Bn + 1 - (m - 1)(n - 1)$.
\end{lemma}

\begin{proof}
    By Remark~\ref{rem:reflection}, for every value in $H_o$ its reflection will also be in $H$.
    That being the case, with $A \geq n$ and $B \geq m$, it holds that $\forall v_o \in H_o$
    \begin{equation*}
        v_o \leq g < mn \leq \frac{1}{2}(Am + Bn) \text{,}
    \end{equation*}
    implying also the opposite for the reflection of $v_o$
    \begin{equation*}
        (Am + Bn) - v_o > \frac{1}{2}(Am + Bn) \text{,}
    \end{equation*}
    then $H_o$ and the set of its reflections are disjoint and it follows that $\abs{H} \geq 2\abs{H_o} = 2h = (m - 1)(n - 1)$.
    Finally, from \eqref{def:H} we have $\abs{S} = \abs{\{0, \ldots, Am + Bn\}} - \abs{H}$ since $S \subseteq \{0, \ldots, Am + Bn\}$, consequently
    \begin{equation*}
        \abs{S} = Am + Bn + 1 - \abs{H} \leq Am + Bn + 1 - (m - 1)(n - 1).
    \end{equation*}
\end{proof}

What will be proven next is that no other value belongs to $H$ except for those in $H_o$ and their reflections, leading to $\abs{H} = (m - 1)(n - 1)$. In turn implying $\abs{S} \leq Am+Bn - 2h$.

\begin{lemma}\label{lemma:no_one_missed}
    Let $s \in \left\{0, \ldots, \frac{1}{2}(Am + Bn)\right\} \smallsetminus H_o$, with $H_o$ as in \eqref{def:H_o}. Then $s \in S$.
\end{lemma}

\begin{proof}
    Since $s \notin H_o$, $\exists (a, b) \text{ s.t. } s = am + bn$.
    If such pair of coefficients is valid, that is $(a, b) \in I$, the proof is complete.
    Suppose this is not the case, and, without loss of generality, assume $a > A$.
    Then
    \begin{equation*}
        Am + bn < am + bn = s \leq \frac{1}{2}(Am + Bn)\text{,}
    \end{equation*}
    thus
    \begin{equation*}
        m < (B - 2b)\frac{n}{A} \leq B - 2b \leq B - b\text{,}
    \end{equation*}
    and therefore
    \begin{equation*}
        b + m < B\text{.}
    \end{equation*}
    
    If we then rewrite $s$ by adding and subtracting $nm$ we get:
    \begin{equation*}
        s = am + bn = (a - n)m + (b - m)n \text{,}
    \end{equation*}
    where having $a > a - n > 0$ and $0 < b + m < B$ makes the pair $((a - n), (b - m))\in I$ valid.
    
    Therefore, by iterating this procedure with:
    \begin{equation*}
        \begin{cases}
            a \leftarrow a - n \\
            b \leftarrow b + m
        \end{cases}
    \end{equation*}
    we get a sequence of coefficients for which $b$ is always admissible and $a$ is eventually admissible.
    Hence, after a finite amount of iterations we get:
    \begin{equation*}
        (a^*, b^*) \in I \text{ s.t. } s = a^*m + b^*n
    \end{equation*}
    Consequently, $s \in S$.
\end{proof}


\begin{remark}\label{rem:no_one_missed}
    Let $s \in \left\{\frac{1}{2}(Am + Bn) + 1, \ldots, (Am + Bn)\right\} \smallsetminus \{(Am + Bn) - v_o \colon v_o \in H_o\}$, with the latter set containing the reflections of the values in $H_o$ as per \eqref{def:H_o} and Remark~\ref{rem:reflection}. Then $s \in S$.
\end{remark}

\begin{proof}
    From Lemma~\ref{lemma:no_one_missed} and Remark~\ref{rem:reflection} it follows that for any $s$ as stated above, if $s$ is a reflection of a value in $H_o$, then $s \in H$, otherwise it must be that $s \in S$ because its reflection is in $S$.
\end{proof}

\begin{lemma}\label{lemma:large_coeffs}
    If both $A \geq n$ and $B \geq m$ then $\abs{S} = Am + Bn + 1 - (m - 1)(n - 1)$.
\end{lemma}

\begin{proof}
    Lemma~\ref{lemma:no_one_missed} and Remark~\ref{rem:no_one_missed} imply that $H = H_o \cup \{(Am + Bn) - v_o \colon v_o \in H_o\}$, from which follows $\abs{H} \leq 2\abs{H_o} = 2h = (m - 1)(n - 1)$.
    Combining this with $\abs{H} \geq (m - 1)(n - 1)$ from Lemma~\ref{lemma:large_coeffs_like_original}, and knowing that $\abs{S} = \abs{\{0, \ldots, Am + Bn\}} - \abs{H}$ from \eqref{def:H}, results in the stated equality.
\end{proof}

Finally, Proposition \ref{prop:main_result} trivially follows from Lemmata \ref{lemma:small_coeffs} and \ref{lemma:large_coeffs}.

\printbibliography

\end{document}